\documentclass[conference]{IEEEtran}
\IEEEoverridecommandlockouts
\usepackage{amsthm}

\newtheorem{thm}{Theorem}[section]
\newtheorem{lem}[thm]{Lemma}
\newtheorem{prop}[thm]{Proposition}

\newtheorem{rem}{Remark}

\usepackage{quantikz}
\usepackage{float}
\usepackage{url}
\makeatletter
\let\newfloat\newfloat@ltx
\makeatother
\usepackage{algcompatible}
\usepackage{amsmath,amssymb}
\usepackage{comment}
\usepackage{eso-pic}

\def\BibTeX{{\rm B\kern-.05em{\sc i\kern-.025em b}\kern-.08em
    T\kern-.1667em\lower.7ex\hbox{E}\kern-.125emX}}
\usepackage{cite}

\begin{document}

\AddToShipoutPictureBG*{%
  \AtPageLowerLeft{%
    \put(0,25){\parbox{\paperwidth}{\centering\footnotesize
       \copyright~2026 IEEE. Accepted at IEEE QAI 2026. Personal use of this material is permitted. Permission from IEEE must be obtained for all other uses.}}}}

\title{
Verifiable Sampling for Delegated Generative Quantum Machine Learning
}

\author{\IEEEauthorblockN{Matteo Inajetovic}
\IEEEauthorblockA{\textit{Electrical Engineering and}\\
\textit{Computer Science Department}\\
\textit{Technische Universit\"at Berlin}\\
Berlin, Germany \\
m.inajetovic@tu-berlin.de}
\and
\IEEEauthorblockN{Anna Pappa}
\IEEEauthorblockA{\textit{Electrical Engineering and}\\
\textit{Computer Science Department}\\
\textit{Technische Universit\"at Berlin}\\
Berlin, Germany \\
anna.pappa@tu-berlin.de}
\and
\IEEEauthorblockN{Petros Wallden}
\IEEEauthorblockA{\textit{Quantum Software Lab}\\
\textit{School of Informatics}\\
\textit{University of Edinburgh}\\
Edinburgh, United Kingdom \\
petros.wallden@ed.ac.uk}
}

\maketitle

\begin{abstract}
    Generative machine learning has shown remarkable results across a wide range of interdisciplinary tasks, motivating significant efforts toward Quantum Generative Models as a path to beyond-classical performance, given their capability of generating classically intractable distributions. However, comparatively little attention has been paid to securely delegating such computations to a powerful remote quantum server. 
    In this work, we propose a blind and verifiable protocol for sampling that can be used for the inference phase of such quantum models. Such a protocol inherits blindness from Universal Blind Quantum Computing and achieves verifiability by exploiting known approaches based on trapification schemes. 
    Alongside, to delegate the training phase, we discuss an extension of the previously proposed verifiable end-to-end Delegated Variational Quantum Algorithms (Inajetovic et al., Phys. Rev. Research, 2026).
    Beyond the usual Accept/Abort outcome characterizing verification protocols, the sampling protocol for inference also provides, even against a malicious server, an estimation of the amount of corrupted generated samples, i.e., a fine-grained feature describing the quality of the outcome connecting verification protocols and application-level benchmarking.

\end{abstract}


\section{Introduction}

Quantum Generative Models (QGMs) have been proposed to learn classical and quantum target distributions \cite{PhysRevA.98.062324, Benedetti_2019} motivated by the success of classical generative models in deep learning and by potential theoretical advantages \cite{pirnay2024unconditionaldistributionlearningadvantage,gyurik2024exponentialseparationsclassicalquantum}.
Rudolph et al.~\cite{Rudolph_2024} analyzed the challenges related to variational QGMs, investigating trainability issues and providing potential improvements. Furthermore, recent works explored new methods and approaches to tackle training limitations and get closer to useful QGMs applications \cite{Coopmans_2024,Mart_nez_de_Lejarza_2025,Demidik_2025}.  
Given the notable progress in the QML's field of dequantization, recent works have highlighted quantum advantages in scenarios where the training can be simulated efficiently classically \cite{huang2025generativequantumadvantageclassical,recioarmengol2025trainclassicaldeployquantum,kolarovszki2026generativemodelinggaussianboson}. 
Nevertheless, such train-classical deploy-quantum strategies have been shown to have limitations, namely discrepancies when deploying quantumly a model trained previously classically for certain scenarios \cite{herrerogonzalez2025bornultimatumconditionsclassical}.
In the context of measurement-based quantum computing (MBQC) \cite{raussendorf_one-way_2001} a QGM exploiting the probabilistic nature of MBQC has been proposed by Majumder et al.~\cite{Majumder_2024}.
Although there is a huge effort in investigating QGMs, the same does not hold concerning protocols to securely delegate these models to a remote quantum server. 
Delegated quantum computation provides a framework in which a computationally weak client can securely delegate a quantum computation to a more powerful quantum server while ensuring guarantees such as correctness, blindness, and, in some settings, verifiability. MBQC-based foundational protocols, such as universal blind quantum computation (UBQC) \cite{ubqc} establish information-theoretic blindness, while later works introduce verifiability \cite{UVBQC17} and robustness against noise for BQP computations \cite{vbqp}. More recent developments extend the delegated paradigm to finer-grained tasks, such as verifiable end-to-end delegated variational quantum algorithms (DVQAs) \cite{dvqas} and composably secure protocols for delegated estimation of expectation values \cite{yang2025verifiableblindobservableestimation}, which are particularly relevant for near-term quantum applications.
Additionally, Chung et al.~\cite{10.1007/978-3-031-07082-2_25} provide a verifiable and blind sampling protocol that could be used to delegate quantum generative models. Nevertheless, it is based on the known Mahadev protocol \cite{8555111}, hence its security is based on computational assumptions rather than information-theoretical ones and it is more complicated to run on near-term devices.
Accreditation protocols bound the error of the outputs of noisy quantum devices using trap circuits~\cite{Ferracin_2019,Mills_2026}, but assume an honest device whose noise satisfies suitable modelling assumptions, whereas our setting allows for an adversarial server.
In this work we provide a general verifiable delegated sampling protocol, which we apply to the inference phase of QGMs. For training, whereas verifiable end-to-end DVQAs are limited to expectation-value losses with Parameter-Shift-Rule gradients, we discuss a generalization of verifiable training to QGMs beyond the DVQA requirements, covering the full pipeline from training to inference.
To support our findings, we also perform numerical simulations  using Veriphix \cite{Abdul_Sater_Veriphix}, a library based on the measurement-based quantum software Graphix \cite{graphix26}.


\section{Background}
Two main features are essential to delegate quantum computations, namely blindness and verifiability.
The former allows for input, circuit, and output privacy, while the latter ensures the client that the computation carried out by the potentially malicious server has been done correctly, returning a low probability to accept a wrong computation.
In the DVQAs framework \cite{dvqas}, blindness is inherited directly from UBQC \cite{ubqc}, ensuring that all delegated computations hide inputs, intermediate angles, and outputs from the server.
Verifiability instead is achieved through a step-level verification protocol ensuring correctness of the gradient updates and end-to-end verifiability based on convergence guarantees for the optimization. In detail, the Verifiable Step Protocol verifies each gradient update via:

\begin{enumerate}
    \item \textit{Bounding corruption via expectation-value deviation} - 
    Establishing how corrupted measurement rounds affect the Parameter-Shift-Rule \cite{Mitarai_2018} gradient estimation yields a bound connecting the number of corrupted rounds to the relative gradient error (Lemma III.1 in \cite{dvqas}).
    This connection allows one to define a failure probability, for the case when a wrong computation is accepted.
    \item \textit{Trap-based verification of gradient delegation} -
    A trapification protocol interleaving computation and test rounds detects adversarial behaviour with failure probability decreasing exponentially in the number of rounds (Theorem III.2 in \cite{dvqas}).
\end{enumerate}
The Verifiable Optimization Protocol then runs a custom gradient-descent optimizer (Protocol \ref{prot:Opt}) that delegates each step through the Verifiable Step Protocol. Choosing appropriate learning rate and error thresholds guarantees geometric convergence, which in turn ensures full-algorithm verifiability (Theorem III.3 in \cite{dvqas}).

\vspace{-.2cm}
\noindent\rule{\linewidth}{0.4pt}
\vspace{-.7cm}
\begin{center}
    Protocol 1 : \textbf{Verifiable Optimization (from \cite{dvqas})}
\end{center}
\vspace{-.5cm}
\noindent\rule{\linewidth}{0.4pt}


\begin{algorithmic}[1]
 \renewcommand{\algorithmicrequire}{\textbf{Input:}}
 \renewcommand{\algorithmicensure}{\textbf{Protocol:}}

 \REQUIRE Number of optimization steps $N_{steps}$, learning rate $\alpha$, error threshold $\varepsilon_{th}$.  
 \ENSURE  
  \FOR {$k$ in $N_{steps}$}:   \label{prot:Opt}
  \STATE Delegate the Verifiable Step Protocol that gets $\hat g(\theta^k)$ and the output signal:
  \IF{Accepts}
  \STATE $\theta^{k+1}=\theta^{k}-\alpha \hat{g}(\theta^{k})$.
  \ELSIF {Aborts}
  \STATE $\theta^{k+1}=\theta^{k}$ (i.e. re-run the same step).
  \ENDIF
  \ENDFOR
\vspace{.1cm}
\end{algorithmic} 
\vspace{-.4cm}
\noindent\rule{\linewidth}{0.4pt}

In this work, building on \cite{dvqas}, we tackle training and inference phases for variational generative modeling algorithms. We address the former by using the circuit's output probability distribution instead of expectation values for the gradient approximation, and the latter by extending the two steps of the Verifiable Step Protocol to the circuit's sampling output.

\section{Methods}

\subsection{Delegated Sampling}

In order to delegate the inference phase of a QGM we build on top of the same trapification scheme used in \cite{vbqp,dvqas,yang2025verifiableblindobservableestimation}.
Let $\mathcal{C}$ be a quantum circuit producing as outcome a probability distribution $p$ with $d$ shots and $\delta$ corrupted shots.
While standard VBQC \cite{UVBQC17} could in principle be applied per shot, this would be significantly more resource-demanding; our trapification-based approach amortizes the verification overhead across rounds, as discussed in \cite{dvqas}.
Given $\gamma=\delta/d$ then  $\hat p = (1-\gamma)p + q\gamma$, where $q$ represents the outcome probabilities of the corrupted shots. Hence the loss function with respect to the non-corrupted scenario will be:
    \[
    \mathcal{L}(p,\hat p)=\mathcal{L}(p, (1-\gamma)p + q\gamma)\leq h_{\mathcal{L}}(\gamma)
    \]
where $h_{\mathcal{L}}$ is a function depending on the ratio of the corrupted shots.
When choosing $\mathcal{L}$ s.t. $\mathcal{L}(p,q)=0$ only if $p=q$, then
\[
h_{\mathcal{L}}(\gamma) = \gamma\,\mathcal{L}(p,q).
\]
This condition holds when $\mathcal{L}$ is the Maximum Mean Discrepancy (MMD)
\cite{MMD12} or the Total Variation Distance (TVD), but not in general: the
Kullback--Leibler divergence $D_{KL}$, for instance, involves a logarithm of
the ratio $p/\hat{p}$, which does not factorize linearly in $\gamma$ and
therefore requires a separate bound. Our goal is to analyze the worst-case
distribution $q$ and find the corresponding upper bound for
$h_{\mathcal{L}}$. The latter enables us to determine whether a computation
can be declared wrong due to the amount of corrupted samples.

\begin{lem} Let $\mathcal{L}$ be the Maximum Mean Discrepancy $\text{MMD}_k$ with a bounded kernel $k$. In presence of $\delta$ corrupted shots over $d$ total shots, the absolute difference between the ideal loss and the corrupted one is bounded as follows:
    \[|\mathcal{L}(p,\hat p)|\leq 2\sqrt K\gamma. \]
where $K$ is the kernel bound.
\label{lemma:mmd}
\end{lem}
 
\begin{proof}
   From \cite{MMD12} and the triangle inequality,
   \[
   \text{MMD}_k(p,q)=\|\mu_p-\mu_q\|_{\mathcal{H}}\leq\|\mu_p\|_{\mathcal{H}}+\|\mu_q\|_{\mathcal{H}}.
   \]

Assuming $k(x,x)\leq K<\infty \text{ for all } x$ we have that:
    \begin{align*}
         \mathbf{E}_{x,x'}[k(x,x')]&=\mathbf{E}_{x,x'}[\langle \phi(x),\phi(x')\rangle_{\mathcal{H}}]\\&\leq \mathbf{E}_{x,x'}[\sqrt{k(x,x)k(x',x')}] \leq K
    \end{align*}
    where the bound comes from the Cauchy-Schwarz inequality.
Then:
\begin{align*}
    \|\mu_p\|^2_{\mathcal{H}}&=\langle\mu_p,\mu_p\rangle_{\mathcal{H}}\\&=\mathbf{E}_{x,x'}[\langle\phi(x),\phi(x')\rangle_{\mathcal{H}}]=\mathbf{E}_{x,x'}[k(x,x')]\leq K
    \end{align*}
       Therefore $ \text{MMD}_k(p,q)\leq 2\sqrt K$ and since $\mu_{\hat{p}} = (1-\gamma)\mu_p + \gamma\mu_q$ by linearity of expectation, we have $\text{MMD}_k(p,\hat{p}) = \gamma \text{MMD}_k(p,q)$, and hence:
    \begin{align*}
       |\text{MMD}_k(p,\hat p)| &\leq\gamma\text{MMD}_k(p,q)\leq 2\sqrt K\gamma.
    \end{align*} 
\end{proof}
\vspace{-.8cm}

\begin{rem}
Analogous bounds hold for other common distance measures. 
If we consider the Total Variation distance  $\text{TVD}(r,s)=\frac{1}{2}||r-s||_{1}$, then
$\text{TVD}(p,\hat p)=\gamma\,\text{TVD}(p,q)\leq\gamma$, by linearity and using
$\text{TVD}(p,q)\leq 1$.
For the Kullback--Leibler divergence $D_{KL}$ the argument
differs. Given $D_{KL}(a,b)=\sum_x a(x)\log\frac{a(x)}{b(x)}$ and $\hat p\geq(1-\gamma)p$, we get $D_{KL}(p,\hat p)\leq-\log(1-\gamma)$. In all three cases $h_{\mathcal{L}}(\gamma)\to 0$ as $\gamma\to 0$, which is the only property
required by Protocol~\ref{prot:Sampling} and Theorem~\ref{th:ver_sampl}.
\end{rem}
\begin{rem}
    The bound in the previous Lemma characterizes the bias introduced by corrupted rounds and is saturated in the limit of infinitely many shots. In practice, estimating $\hat{p}$ from $d$ finite shots introduces an additional variance term for the distance measures. Throughout the analysis we treat $d$ as sufficiently large so that this shot-noise contribution is negligible relative to $\varepsilon_{th}$, while  
a complete treatment requires absorbing it into the bound:
\[
 h'_{\mathcal{L}}(\gamma) :=  h_{\mathcal{L}}(\gamma)+\text{S}_{\mathcal{L}}(d)
\]
where $\text{S}_{\mathcal{L}}(d)$ denotes the distance error due to shot-noise.
\end{rem}

\textbf{Protocol}
In case we have a desired error threshold $\varepsilon_{th}$, then we set $\varepsilon_{th}=h_{\mathcal{L}}(\gamma_{max})$ where $\gamma_{max}$ is the maximum tolerable ratio of corrupted shots. We propose the following protocol:
\vspace{-.1cm}

\noindent\rule{\linewidth}{0.4pt}
\vspace{-.7cm}
\begin{center}
    Protocol 2 : \textbf{Verifiable Sampling}
\end{center}
\vspace{-.5cm}
\noindent\rule{\linewidth}{0.4pt}

 \begin{algorithmic}[1]
 \renewcommand{\algorithmicrequire}{\textbf{Input:}}
 \renewcommand{\algorithmicensure}{\textbf{Protocol:}}

 \REQUIRE Computational graph $G$ of a quantum circuit $\mathcal{C}$, associated $c$-coloring, distance function $\mathcal{L}$ and error threshold target $\varepsilon_{th}$.

 \ENSURE  
 
 \\ \textit{Initialization} : Suppose $n=d+t$ the number of rounds, where $d$ is the number of computation rounds and $t$ the number of test rounds. The client chooses at random which round belongs in each type.   
\\ \textit{Delegation} \label{prot:Sampling}:
  \FOR {$i$ in $n$}: 
  \IF {$i \in T$}
  \STATE the client sends the trap and dummy qubits.
  \ELSE
  \STATE the client sends the qubits associated with $G$.
  \ENDIF
  \STATE The server performs UBQC and returns the output.
  \ENDFOR
  \\ \textit{Outcome :}
  \STATE The client gathers and decodes the server's outputs:
  \IF {there are trap failures}:
  \STATE Store the amount of trap failures ($N_T$) and \textbf{Abort}
  \ELSE
  \STATE \textbf{Accepts} and obtains $\hat p$ with $|\mathcal{L}(\hat p, p)|<\varepsilon_{th}$.
  \ENDIF
\vspace{.1cm}
\end{algorithmic} 
\vspace{-.4cm}
\noindent\rule{\linewidth}{0.4pt}

The following theorem provides verifiability for the sampling protocol.

\begin{thm}[Verifiable Sampling]
 Let $\mathcal{L}$ be a distance measure s.t. $|\mathcal{L}(\hat p, p)|\leq h_{\mathcal{L}}(\gamma)$. Given a quantum circuit $\mathcal{C}$ with output probability distribution $p$ in the honest scenario, Protocol \ref{prot:Sampling} samples an output probability distribution with absolute error up to $\varepsilon_{th}=h_{\mathcal{L}}(\gamma_{max})$ and probability of failure decreasing exponentially with the number of rounds $n$. 
\label{th:ver_sampl}
\end{thm}
\begin{proof}[Sketch Proof]
The proof is similar to the one of Theorem III.2 in \cite{dvqas}; the main difference is the way one obtains the ratio between the allowed corrupted computation rounds $\delta_{max}$ and the total computation rounds $d$, namely $\gamma_{max}$. 
The outcome of Protocol \ref{prot:Sampling} is the output probability distribution obtained by sampling the quantum circuit $\mathcal{C}$. Therefore, the aforementioned ratio is derived from the error $|\mathcal{L}(\hat p, p)| \leq h_{\mathcal{L}}(\gamma_{max})=\varepsilon_{th}$ where having $\delta\leq\delta_{max}$ implies
$|\mathcal{L}(\hat p, p)|\leq\varepsilon_{th}$ and a computation is declared
wrong whenever $\delta>\delta_{max}$.
Given this modification, we can obtain an upper bound on the probability of failure:
\begin{align*}
\text{Pr}(\text{fail}) &\leq  \max_{m \in [0, n]}
\text{Pr}[ |\mathcal{L}(\hat p, p)| \geq \varepsilon_{th} \text{ } \land \text{ Client Accepts}]\\&\leq 
\max_{m \in [0, n]}\text{Pr}[X_d \geq \delta_{max} \land X_t < 1] 
\label{eq:prob_fail}
\end{align*}

where $m$ is the number of attacked rounds, and $X_d$ and $X_t$ are the random variables counting the attacked computation rounds and the attacks detected by the traps, respectively.

From this point on, showing that $\text{Pr}(\text{fail})$ decreases exponentially with the number of rounds $n$ follows directly from the proof in \cite{dvqas}.
\end{proof}

In order to provide verifiability we set an error threshold over which we declare the outcome to be wrong. Additionally, we use the number of failures in the test rounds as a diagnosis for a more fine-grained analysis of the outcome's quality.
Since the ``Accept/Abort" decision is taken after the delegation of all the rounds, we underline that the protocol, independently of the final decision, can also return an estimation of the error. This is inherently given by the statistical relation between the number of trap detections in the test rounds and the corrupted computational rounds. 
The following proposition provides a quality measure for the outcome, analogous to accreditation bounds~\cite{Ferracin_2019,Mills_2026} but holding in the adversarial delegated setting, and strengthening the link between the verification protocol and potential application-level benchmarking.

\begin{prop}
    The confidence interval $\left[\gamma_{l},\gamma_u\right]$ for the rate of corrupted computation rounds contains the true value with probability at least $1-\alpha$, with a margin of error $c\sqrt{\frac{1}{2t}\ln(2/\alpha)}$.
     \label{prop:est}
\end{prop}

\begin{proof}
    Given $\delta \approx \frac{md}{n}$ and $\gamma=\frac{\delta}{d}$, 
we can rewrite the number of attacks detected by the traps as $X_t \approx \text{Binomial}(t,\frac{\gamma}{c})$. This holds because blindness prevents the server from distinguishing test
from computation rounds, so the $m$ attacked rounds are spread uniformly over
the $n$ rounds: each test round is then attacked with probability $\gamma$
and detected with probability $1/c$, independently across rounds.
Using Hoeffding's inequality with $\eta>0$, we have:

\[
\text{Pr}\left(\left| X_t - \frac{\gamma t}{c} \right| \geq \eta t\right)\leq 2 \exp(-2\eta^2 t)=\alpha
\]

Hence the error can be expressed as: $\eta = \sqrt{\frac{1}{2t}\ln(2/\alpha)}$ and
$\gamma\in\left[c\left(\frac{X_t}{t}-\eta\right),c\left(\frac{X_t}{t}+\eta\right)\right] \text{ with probability at least } 1-\alpha$.

\end{proof}
\vspace{-.2cm}
In the following subsection, we address the training phase of QGMs in the delegated setting.

\subsection{Delegated End-to-End Generative QML}
The DVQAs framework \cite{dvqas} works under two structural conditions: loss functions are expressible as expectation values and gradients are estimated via the Parameter-Shift Rule~\cite{Mitarai_2018}. Quantum Generative Models could violate both of these conditions: in general, sampling-based distance measures do not always reduce to expectation values directly, and gradient estimation methods different than the Parameter-Shift Rule are used across different QGMs. This structural mismatch means DVQAs cannot be applied out-of-the-box to securely delegate the training of QGMs.
Hence, rather than focusing on every specific QGM and providing related detailed theorems re-adapting the DVQAs optimization, we propose a more general framework to address secure delegation of QGMs.

We define $\mathcal{P}(\mathcal{L}, \mathcal C, \theta)$ to be a  gradient estimation method where $\mathcal{L}$ is the loss function, $\mathcal C$ the ansatz and $\theta$ the parameters of the ansatz.
We identify two looser conditions under which verifiable optimization remains achievable despite the aforementioned mismatch with the original DVQAs.
\begin{enumerate}
    \item The error introduced by corrupted rounds in the loss function $\mathcal{L}$ can be bounded as 
    a function of $\gamma$.
    \item All circuit evaluations within a gradient step defined by $\mathcal{P}$ share the 
    same MBQC graph, enabling blind shared-trap delegation.
\end{enumerate}
The first condition ensures that, independently of the loss function being used, it is possible to connect the number of corrupted rounds with the computation's output quality.
The second one ensures that, during a gradient step delegation, test and computation rounds remain indistinguishable to the server. Rounds with different graphs
would reveal structural information and consequently preclude shared delegation in the same protocol execution.
When both conditions hold, we can apply the following workflow that generalizes previous work \cite{dvqas}:
\begin{itemize}
    \item Verifiable Step:
    \begin{enumerate}
    \item Define an error measure for the gradient estimation method $\text{Err}(\mathcal{P})$ and
    provide an upper bound $h_{\mathcal P}(\gamma)$ depending on the corrupted computation rounds ratio $\gamma=\delta/d$ s.t.  
    \[\text{Err}(\mathcal{P}(\mathcal{L}, \mathcal C, \theta))< h_{\mathcal P}(\gamma)\] 
    \item Given a target error threshold $\varepsilon_{th}=h_{\mathcal P}(\gamma_{max})$, define a wrong computation when the corrupted rounds exceed a tolerable amount $\delta_{max}$ :
    \[\text{if } \delta>\delta_{max} \text{ then Err}(\mathcal{P})\geq\varepsilon_{th}\]
    \item Using the aforementioned trapification scheme \cite{dvqas,yang2025verifiableblindobservableestimation,vbqp} allows one to bound the probability of failure: 
        \begin{align*}
        \centering
            \text{Pr}[ \text{Client Accepts } \land \text{ Err}(\mathcal{P}) \geq \varepsilon_{th}] \leq \epsilon
        \end{align*}
\end{enumerate}

\item Verifiable Optimization: 
As done for Theorem III.3 in \cite{dvqas},
by looking at the convergence conditions of gradient descent, it is possible to determine $\varepsilon_{th}$ and therefore run the custom gradient descent optimization Protocol \ref{prot:Opt}, delegating iteratively the new Verifiable Step.
\end{itemize}
Together, the three requirements within the Verifiable Step above recover what the Verifiable Step
Protocol of \cite{dvqas} provides for the Parameter-Shift Rule: a threshold separating acceptable from wrong gradient estimates, and a
delegation protocol whose probability of accepting a wrong one decreases
exponentially in the number of rounds.
Note that some works use different circuits to estimate the gradients (e.g. in \cite{demidik2025samplebasedtrainingquantumgenerative}) violating the second condition. A potential solution would be to run independent calls to Protocol~\ref{prot:Sampling} for each circuit evaluation, albeit at the cost of losing the shared-trap efficiency and requiring a revised end-to-end convergence argument.
We provide here two examples of delegating training and inference of two QGMs where our hypothesis on sharable traps holds: the Quantum Circuit Born Machine (QCBM) using the MMD loss and the MBQC-based generative model \cite{Majumder_2024}.

\begin{figure*}[t]
    \centering
    \includegraphics[width=.85\linewidth]{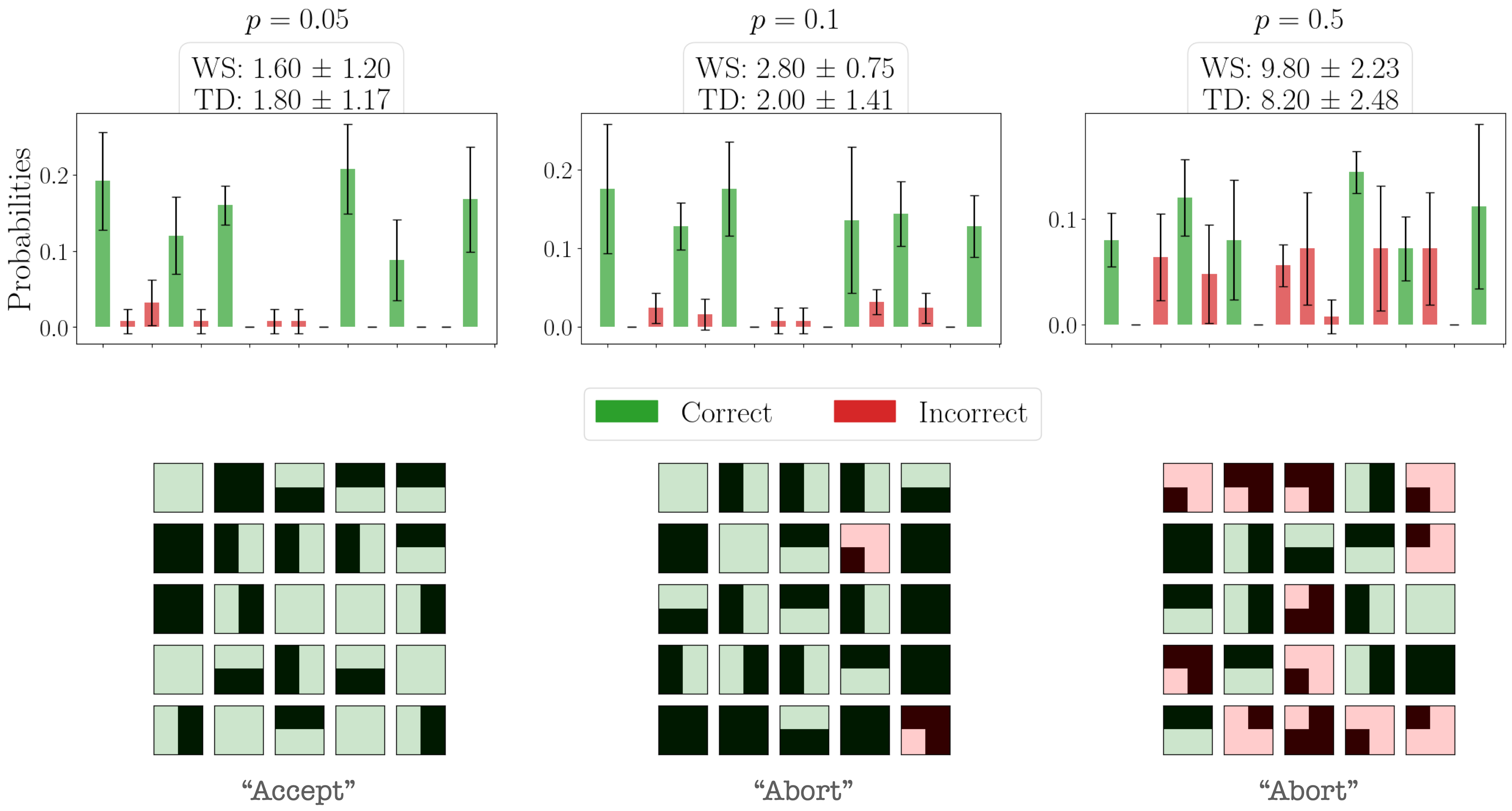}
    \caption{\textit{Upper}: $d=25$, $t=50$, $5$ Repetitions. WS: Number of wrong generated samples. TD: Number of trap detections. $p$: probability to attack a round as defined in the main text. In the x-axis we have the $2^n$ bitstrings samples where $n$ is the number of bits.
    In green we highlight the correct samples belonging to the target distribution while in red the wrong ones. \textit{Bottom}: Visualization of one of the runs from the upper figure experiment.}
    \label{fig:dqi}
\end{figure*}

\paragraph*{Example 1: QCBM+MMD}

The verifiable training follows directly from DVQAs, since the MMD gradient can be expressed as expectation values via \cite{Rudolph_2024}, making the Parameter-Shift Rule applicable with fixed circuit structure.
For the inference phase instead, Protocol~\ref{prot:Sampling} is required.

\paragraph*{Example 2: MBQC-based QGM}
For the MBQC-based generative model of \cite{Majumder_2024},  Protocol~\ref{prot:Sampling} applies directly to the inference phase. The client delegates sampling with verifiability guarantees given by Theorem~\ref{th:ver_sampl}.
Concerning training, unlike Example 1, the MMD gradient here is estimated directly from samples via Parameter-Shift Rule, without 
reduction to expectation values. This falls outside 
the direct application of DVQAs' end-to-end training and requires a new Verifiable Step.
Let $g$ be the gradient vector and $K$ the kernel bound for the MMD loss $\mathcal{L}$. We can estimate each gradient component as follows:
\[
g_i=\frac{\partial \mathcal{L}}{\partial \theta_i}
=\frac{1}{2}(\mathcal{L}_{i,+}-\mathcal{L}_{i,-})
\]
where $\mathcal{L}_{i,+}$ and $\mathcal{L}_{i,-}$ correspond to the two shifted evaluations of the loss function.
By applying Lemma~\ref{lemma:mmd} to each gradient component and given $N_P$ the number of parameters of the circuit (and consequently the dimension of the gradient vector):

\[
|g_i - \hat{g}_i| \leq \frac{2\sqrt{K}\gamma_+^i + 
2\sqrt{K}\gamma_-^i}{2}= \sqrt{K}\gamma_i
\]
where $\gamma_\pm^i$ are the corrupted shot ratios for the 
two function evaluations and $\gamma_i = \gamma_+^i + 
\gamma_-^i$. Aggregating over all parameters
and defining $\gamma$ as the global corrupted shot ratio 
across all $2N_P$ evaluations, we get:
\[
\|g - \hat{g}\|_2 \leq 
\sqrt{KN_P}\,\gamma
\]
The relative error is defined as $\varepsilon = \|g-\hat{g}\|_2 / \|g\|_2$. Given $\epsilon_0$ such that $\|g\|_2 \geq \epsilon_0  > 0$ we obtain 
$h_{\mathcal{P}}(\gamma) =  \sqrt{KN_P}\gamma / \epsilon_0$.
Given a desired error threshold $\varepsilon_{th} = h_{\mathcal{P}}(\gamma_{max})$, we achieve a verifiable delegated gradient estimation using trapification. Hence,
end-to-end verifiable training follows from using this new verifiable step within step 2 in
Protocol~\ref{prot:Opt}.

\section{Numerical Simulations}

In this section, we present numerical simulations for the inference of QGMs employing the Verifiable Sampling Protocol \ref{prot:Sampling}. Concerning the training phase, we have provided a high-level approach to extend what has already been shown in \cite{dvqas} and thus in this section, we focus only on inference.
At the inference phase, the goal of running the verifiable sampling protocol is not only to reproduce or compare the full probability distribution with the target one. In our framework the client can also be confident that, whenever no trap fails, the generated samples are computed with the chosen ansatz and follow the attack-free target distribution learned during training, up to a threshold error $\varepsilon_{th}$.
We tackle the Bars\&Stripes Problem, a common benchmark for quantum generative models, by solving the $2\times2$ problem instance with a 4-qubit Instantaneous Quantum Polynomial (IQP) Born Machine based on \cite{recioarmengol2025trainclassicaldeployquantum}.
Given the optimal angles obtained classically via \textit{IQPopt}, \cite{iqpopt} we 
run the delegated inference phase simulation with Veriphix \cite{Abdul_Sater_Veriphix}, a library for Verification and Benchmarking based on the MBQC software Graphix \cite{graphix26}. 
To simulate deviations from the ideal protocol, namely the attacks, 
we employ the \texttt{MaliciousNoiseModel} from Veriphix. This model applies dephasing noise to a selected set of qubits, independently in each round with probability $p$; for this experiment we target a single output qubit. It is important to underline that the proposed theoretical work is robust to any kind of attack model on the server side and that we use this simple model for a demonstrative purpose. 
The numerical simulations  in Fig. \ref{fig:dqi} illustrate the theoretical findings as the trap detection enables the client to abort the inference when wrong bitstrings are sampled. Furthermore, the  number of trap detections is correlated to the number of wrong samples (TD and WS respectively in Fig. \ref{fig:dqi}), supporting its use as an estimator for output quality during deployment. 
Beyond the setting of Fig.~\ref{fig:dqi}, we repeat the experiment with larger parameters ($d=500$, $t\in\{100,500\}$, $100$ repetitions) and a wider range of attack probabilities, reported in Table~\ref{tab:sweep_p}. The estimator $c\,X_t/t$
recovers the true corruption rate $p$, while its standard deviation decreases as the
number of test rounds grows.
Hence, this experiment is consistent with Proposition \ref{prop:est} and highlights potential applications of our verification protocol in the quantum benchmarking scenario.

\begin{table}[t]
\centering
\footnotesize
\setlength{\tabcolsep}{3pt}
\caption{Delegated inference under the \texttt{MaliciousNoiseModel}
attack, with $d=500$, $t\in\{100,500\}$ test rounds and a
$c$-coloring of the pattern's graph with $c=3$, over $100$ repetitions. $p$ is the per-round attack probability, $c\,X_t/t$ the
estimator of the corrupted-round ratio of Proposition~\ref{prop:est}, and ``Accept'' the fraction of
runs in which no trap failed. WS$/d$ is the measured fraction of wrong samples, averaged over both runs.}
\label{tab:sweep_p}
\begin{tabular}{cc|cc|cc}
\hline
 & & \multicolumn{2}{c|}{$t=100$} & \multicolumn{2}{c}{$t=500$} \\
$p$ & WS$/d$ & $c\,X_t/t$ & Accept & $c\,X_t/t$ & Accept \\
\hline
0.00 & 0.036 & $0.000\pm0.000$ & 1.00 & $0.000\pm0.000$ & 1.00 \\
0.02 & 0.056 & $0.023\pm0.029$ & 0.50 & $0.020\pm0.010$ & 0.02 \\
0.05 & 0.084 & $0.050\pm0.042$ & 0.21 & $0.050\pm0.017$ & 0.00 \\
0.10 & 0.128 & $0.096\pm0.056$ & 0.03 & $0.102\pm0.022$ & 0.00 \\
0.20 & 0.224 & $0.206\pm0.066$ & 0.00 & $0.205\pm0.037$ & 0.00 \\
0.35 & 0.367 & $0.354\pm0.098$ & 0.00 & $0.356\pm0.046$ & 0.00 \\
\hline
\end{tabular}
\end{table}

\section{Conclusion}
We propose a blind and verifiable protocol for delegated quantum circuit sampling and extend the DVQAs framework to Quantum Generative Models, covering both inference and training phases under information-theoretic security assumptions. Beyond the binary ``Accept/Abort" decision, the protocol provides a concrete estimation of the corrupted sample ratio, which, unlike accreditation-based bounds, does not rely on noise-model assumptions. This strengthens the link between verification and application-level benchmarking. 
The verifiable sampling protocol is general and applies beyond QGMs to any sampling-based quantum computation, and the trap detection rate serves as an operational proxy for output quality even when the target distribution is classically inaccessible. However, the end-to-end verifiable training framework requires that all circuit evaluations within a gradient step share the same MBQC graph; QGMs whose gradient estimation relies on heterogeneous circuits fall outside this assumption and will be addressed in future works.
Other future directions include robustness improvement allowing a tunable abort threshold, applications to benchmarking and exploring multi-client/server settings.

\section*{ACKNOWLEDGMENTS}

P.W. acknowledges support from the Quantum Advantage Turbo Charger grant (UKRI4257), the Quantum Advantage Pathfinder (EP/X026167/1) and the QCI3 Hub (EP/Z53318X/1) all funded by UKRI's EPSRC.  A.P. and M.I. acknowledge support from the DFG (via the Priority Program SPP2514) and the Berlin Quantum Alliance.

\bibliographystyle{IEEEtran}
\bibliography{refs}

\end{document}